\documentclass[10pt]{article}
\makeatletter

\renewcommand\normalsize{%
   \@setfontsize\normalsize\@xpt{14}%
   \abovedisplayskip 10\p@ \@plus2\p@ \@minus5\p@
   \abovedisplayshortskip \z@ \@plus3\p@
   \belowdisplayshortskip 6\p@ \@plus3\p@ \@minus3\p@
   \belowdisplayskip \abovedisplayskip
   \let\@listi\@listI}
\normalsize
\renewcommand\small{%
   \@setfontsize\small\@ixpt{12}%
   \abovedisplayskip 8.5\p@ \@plus3\p@ \@minus4\p@
   \abovedisplayshortskip \z@ \@plus2\p@
   \belowdisplayshortskip 4\p@ \@plus2\p@ \@minus2\p@
   \def\@listi{\leftmargin\leftmargini
               \topsep 4\p@ \@plus2\p@ \@minus2\p@
               \parsep 2\p@ \@plus\p@ \@minus\p@
               \itemsep \parsep}%
   \belowdisplayskip \abovedisplayskip
}
\renewcommand\footnotesize{%
   \@setfontsize\footnotesize\@viiipt{10}%
   \abovedisplayskip 6\p@ \@plus2\p@ \@minus4\p@
   \abovedisplayshortskip \z@ \@plus\p@
   \belowdisplayshortskip 3\p@ \@plus\p@ \@minus2\p@
   \def\@listi{\leftmargin\leftmargini
               \topsep 3\p@ \@plus\p@ \@minus\p@
               \parsep 2\p@ \@plus\p@ \@minus\p@
               \itemsep \parsep}%
   \belowdisplayskip \abovedisplayskip
}
\renewcommand\scriptsize{\@setfontsize\scriptsize\@viipt\@viiipt}
\renewcommand\tiny{\@setfontsize\tiny\@vpt\@vipt}
\renewcommand\large{\@setfontsize\large\@xipt{15}}
\renewcommand\Large{\@setfontsize\Large\@xiipt{16}}
\renewcommand\LARGE{\@setfontsize\LARGE\@xivpt{18}}
\renewcommand\huge{\@setfontsize\huge\@xxpt{30}}
\renewcommand\Huge{\@setfontsize\Huge{24}{36}}
\makeatother

\usepackage{textcase}

\usepackage[nofonts]{dgleich-common}

\usepackage[T1]{fontenc}
\RequirePackage[latin1]{inputenc}
\RequirePackage[scaled=0.8]{helvet}%
\RequirePackage[scaled=0.75]{beramono}%
\def\lstbasicfont{\fontfamily{fvm}\selectfont}

\usepackage[labelfont=bf,labelsep=period,justification=raggedright]{caption}

\usepackage{amssymb}
\usepackage{amsmath}
\usepackage{graphicx}
\usepackage{epstopdf}
\usepackage{subfigure}
\usepackage{multicol}

\graphicspath{{./}{./figures/}}

\usepackage{epstopdf}
\graphicspath{{./}{./figures/}}
\usepackage{dgleich-math}
\usepackage{tabularx}
\usepackage{booktabs}
\usepackage{url}
\usepackage{multicol}
\usepackage{todonotes}
\usepackage{paralist}

\newcolumntype{H}{>{\setbox0=\hbox\bgroup}c<{\egroup}@{}}

\usepackage{dgleich-setup}

\usepackage{amsmath}
\usepackage{arydshln}

\newcommand{\YMN}{\mY_{\!\!MN}}
\newcommand{\PM}{\mP_{\!\!M}}
\newcommand{\QN}{\mQ_{\!N}}

\providecolor{theblue}{RGB}{0,0,180}%
\newcommand{\doi}[1]{\href{http://dx.doi.org/#1}{\nolinkurl{#1}}}

\usenatbib
\usehyperref

\makeatletter
\def\blfootnote{\xdef\@thefnmark{}\@footnotetext}
\makeatother

\title{Multimodal Network Alignment}
\author{Huda Nassar, David F.~Gleich}
\date{}

\begin{document}
\marginnote[600pt]{\fontsize{7}{9}\selectfont
Huda Nassar, Purdue University\\\url{hnassar@purdue.edu}\\\noindent David F. Gleich, Purdue University\\\url{dgleich@purdue.edu}}

\maketitle

\vspace{-3\baselineskip}

\begin{abstract} \small\baselineskip=9pt 
	A multimodal network encodes relationships between the same set of nodes in multiple settings, and network alignment is a powerful tool for transferring information and insight between a pair of networks. We propose a method for multimodal network alignment that computes a matrix which indicates the alignment, but produces the result as a low-rank factorization directly. We then propose new methods to compute approximate maximum weight matchings of low-rank matrices to produce an alignment. We evaluate our approach by applying it on synthetic networks and use it to de-anonymize a multimodal transportation network.
\end{abstract}

\section*{Keywords}
network alignment, low-rank matching, multimodal data

\section{Introduction \textit{\&} Motivation}

Network alignment is a technique to identify related nodes between two distinct networks. The result is a formal alignment or matching between the vertices of each network such that a vertex can only be aligned or matched to a single vertex in the opposite network. This methodology has its roots in the domains of ontology alignment~\cite{hu2008-falcon-ao}, protein-protein network analysis~\cite{kelley2004-pathblast,singh2008-isorank-multi}, social network de-anonymization~\cite{Korula-2014-reconciliation}, and object recognition~\cite{conte2004-graph-matching}. The results of network alignment are often used for hypothetical information transfer. 
For instance, if we understand how a particular experiment disrupts a subnetwork of protein interactions in a mouse and we know the related  interactions in a human, we have evidence for the hypothesis that the experiment would disrupt the mapped interactions in a human. 

The network alignment problem and methods have traditionally been formulated for a pair of networks. Recently, much of the network data emerging in scientific applications and engineering studies is multimodal and contains separate types of relational information between vertices~\cite{Szell-2010-multirelational,Cardillo-2013-multiplex,Battison-2014-multiplex,Nicosia-2015-multiplex,Ni-2014-ranking}. For instance, transportation networks are often described as multimodal when they feature multiple interconnecting but distinct networks, such as when subway lines reach airports and national train systems. In biology, multimodal networks show different types of relationships that can occur between proteins and genes such as protein sequence similarity, co-occurrence in genes, co-occurrence in scientific papers, experimental interaction and more. 

This motivates our present manuscript. We seek to produce effective methods to align a pair of multimodal networks. We will be more specific about our notion of a multimodal network in the following section. For now, consider multimodal networks in biology where proteins have many different types of relationships as described above. We may, for instance, be interested in aligning these networks generated from two species to infer related complexes and systems. In this case, the modes of each multimodal network are shared between the two networks to align. 


We propose a new method that takes in a pair of multimodal networks with the same set of modes and produces an alignment between the vertices encoded in the network with no apriori knowledge of the relationship between the vertices. Network alignment problems are notoriously difficult from a theoretical perspective as they  generalize the subgraph isomorphism problem. In this context, principled heuristic solutions abound~\cite{Mohammadi-2016-tame,singh2008-isorank-multi,conte2004-graph-matching}. Our method is a principled and effective heuristic that generalizes the ideas utilized in the IsoRank method~\cite{singh2008-isorank-multi} in concert with the insights from the network similarity decomposition~\cite{Kollias-2011-netalign}. It is designed for multimodal networks with 10s-1000s of modes with 1,000-100,000 vertices.

\begin{marginfigure}
	\centering
	\includegraphics[scale=0.8]{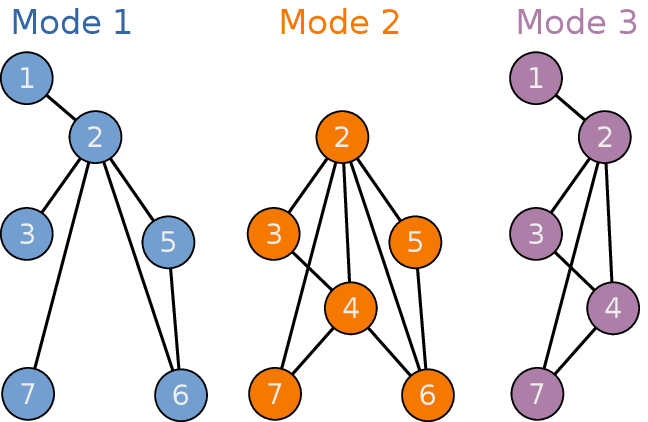}\bigskip
	
	\includegraphics[scale=0.8]{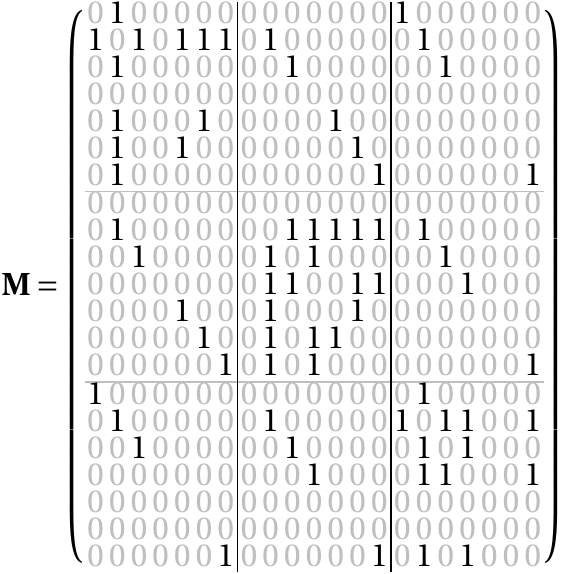}
	\caption{A multimodal network consists of a common set of vertices identifiers (in this case, 1-7) and a set of modes that define edges. Below, we illustrate the multimodal adjacency matrix, which we will use to align the multimodal networks.}
	\label{fig:multimodal}
\end{marginfigure}

\begin{marginfigure}
	\centering
	\includegraphics[scale=0.76]{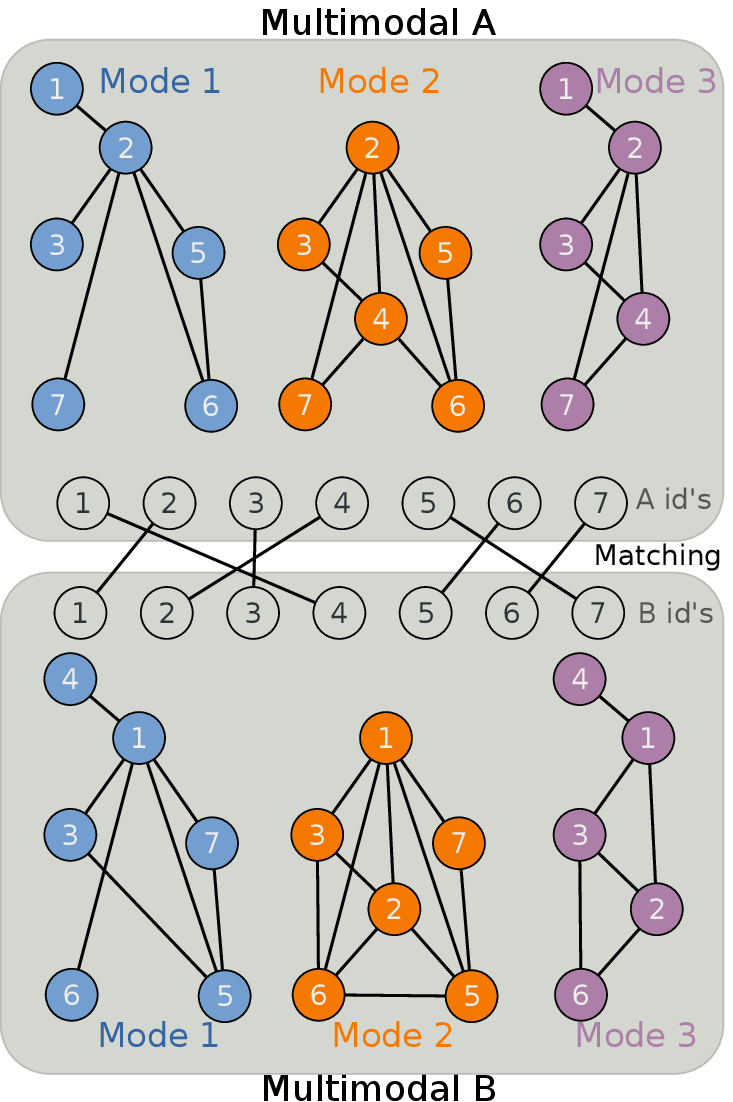}
	\caption{An illustration of the multimodal alignment problem. The goal is to identify the matching illustrated between the node id's given only the multimodal networks $A$ and $B$. We assume the modes have a known alignment. The overlap of this matching is 20. }
	\label{fig:multimodal-alignment}
\end{marginfigure}

On top of this new method, we contribute a new theory about how to take the output from our heuristics and efficiently turn it into a matching between the graphs. This solves one of the key challenges in how to deal with multimodal networks and alignment. Many good ideas result in a polynomial explosion in problem size. This causes methods that have been developed for the pairwise case to be nearly impossible to use for the multimodal setup we propose~\cite{Bayati-2013-netalign,klau2009-network-alignment}. More specifically, this is a memory bottleneck: they would need terabytes of memory to handle problems that start off as a megabyte of data.  We state this new theory as an independently useful primitive of finding a matching in a low-rank matrix. For this problem, we investigate a highly efficient $1/k$ approximation for a rank-$k$ matrix that finds the matching by using the low rank factors of the matrix \textit{only}.

One of the difficulties with network alignment is that effective solutions for applications often exploit and use features of the data arising in those applications, which may not be generalizable to other instances. For instance, in biology, using information about the relationship between the protein sequences is vital to attain the best solutions~\cite{singh2008-isorank-multi,Mohammadi-2016-tame}. Thus, when we seek to evaluate our new methods, we will do so in (i) detailed synthetic experiments and (ii) a case-study and demonstration of our technique where we de-anonymize a publicly released transportation network that lacks meaningful vertex identifiers~\cite{Nicosia-2015-multiplex}. The goal of the synthetic experiments is to address the hypothesis that using multimodal alignment is better than straightforward generalizations of using existing network alignment methods to  pairwise align each mode or a combined network. In the case of the transportation data, only our multimodal alignment method is able to completely map the data between the public data and the original database with labels. Our goal with both of these experiments is to reveal properties of our method and ideas that would be useful in any domain-specific application of network alignment but without the engineering that tends to occur in these applications. Towards that end, we do no tuning or parameter selection.

In short, the contributions of our manuscript are:
\begin{compactenum}
	\item A precise statement of a multimodal network alignment problem (Section~\ref{sec:problem}).
	\item A new method, multimodal similarity decomposition (MSD), to solve multimodal network alignment problems that uses all the information among the modes  (Section~\ref{sec:msd}) and goes beyond baseline methods that align the modes individually (Section~\ref{sec:baseline}).
	\item Independently useful results about how to approximate a bipartite max-weight matching where the matching matrix is low-rank (Section~\ref{sec:matching}). 
	\item We find that simple methods for multimodal alignment outperforms strong methods for pairwise network alignment of the individual modes when there are many vertices missing in each mode (Section~\ref{sec:synthetic}). 
	\item A case-study with de-anonymizing a publicly released multimodal network of airports and airlines~\cite{Nicosia-2015-multiplex} that cannot be de-anonymized using state-of-the-art network alignment methods. We also study the set of most-helpful modes to find this alignment (Section~\ref{sec:airports}). 
\end{compactenum}
We also review existing network alignment methods and ideas to contextualize our contribution (Section~\ref{sec:background}). We will post our codes at \url{www.cs.purdue.edu/~dgleich/codes/multinsd/} for reproducibility and additional sensitivity experiments.


\section{Multimodal networks \textit{\&} Multimodal network alignment}
\label{sec:problem}
\label{multi-creation}
In the context of this paper, a multimodal network is a common set of vertex labels (the set $V$) and multiple sets of undirected edges over these vertices. Each set of edges represents a mode, so we have  $E^{(1)}, \ldots, E^{(m)}$ for $m$ modes.
We illustrate an example in Figure~\ref{fig:multimodal} with three modes. Note that only a subset of vertices may be present in each mode. 

To manipulate multimodal network data computationally, we use a multimodal adjacency matrix. Our definition is in the spirit of how multislice, multiplex networks, and temporal networks are represented~\cite{Mucha-2010-community}, although the details of our specific construction differ. Let $\mA_1, \ldots, \mA_m$ represent the adjacency matrices of each mode. Then the multimodal adjacency matrix is: 
\begin{equation}
 \mM = \bmat{ \mA_1^{} & \mC_{12}^{}  & \ldots & \mC_{1m}^{} \\
	            \mC_{12}^T & \mA_2^{}  & \ldots & \mC_{2m}^{} \\
	            \vdots  & \ddots & \ddots & \vdots \\
	            \mC_{1m}^T & \ldots & \ldots & \mA_m }. 
\label{eqn:multimodalmatrix}
\end{equation}
Here, the matrix $\mC_{ij}$ represents the cross-modal associations between mode $i$ and mode $j$. This is always a binary diagonal matrix where $\mC_{ij}(k,k) = 1$ if vertex $k$ is present in both mode $i$ and mode $j$. All other entries are 0. An example is illustrated in Figure~\ref{fig:multimodal}.  Note that $\mM$ is sparse, and the cross-modal edges do not introduce too many new entries beyond the data. 

The goal of the multimodal network alignment problem is to produce a matching between the vertices of two multimodal networks that \emph{aligns} the networks by maximizing the number of edges of each network that are preserved under the matching. Recall that a matching is a $1-1$ relationship between two sets. Let $V_A$ and $V_B$ be the vertex sets of multimodal networks $A$ and $B$, respectively. We use $v \leftrightarrow v'$ to denote a matching between $v \in V_A$ and $v' \in V_B$. Let $E_A^{(k)}$ and $E_B^{(k)}$ be the edge sets of the $k$th mode of $A$ and $B$ as well. Then we seek a matching  between $V_A$ and $V_B$ that maximizes
\begin{equation}
\sum_{~~k=1~~}^{m} \sum_{\substack{(u_A,v_A) \\ \in E_A^{(k)}}} \sum_{\substack{(u_B,v_B) \\ \in E_B^{(k)}}} \begin{cases} 1 & \substack{\text{if } u_A \leftrightarrow u_B \text{ and } \\ v_A \leftrightarrow v_B} \\ 0 & \substack{\text{otherwise}}. \end{cases}
\end{equation} 
We call this objective the \emph{overlap} of a matching.  
Note that this is a single matching between the vertices that is then evaluated over all the modes, which is possible because we assume that the correspondence between each mode is known and mode $k$ in network A corresponds to mode $k$ in network B. See Figure~\ref{fig:multimodal-alignment} for an example. 

In many cases it is convenient to state the matching as a matrix $\mX$ where the rows are indexed by vertices in $V_A$ and columns by vertices in $V_B$. Let $X(v,v') = 1$ if $v \in V_A$ matches to $v' \in V_B$ and 0 otherwise. Also, let $\{ \mA_k \}$ and $\{ \mB_k \}$ be the adjacency matrices for each mode of the multimodal networks $A$ and $B$. Then the multimodal network alignment problem is
\begin{equation}\label{eq:mm-align}
\MAXthree{\mX}{\tfrac{1}{2}\sum_{k=1}^{m} \sum_{ij} [\mX^T \mA_k \mX]_{ij} [\mB_k]_{ij}}{\sum_{i} X_{ij} \le 1 \text{ for all $j$}}{ \sum_{j} X_{ij} \le 1 \text{ for all $i$}}{X_{ij} \in \{ 0, 1\}.}
\end{equation}
This is an integer optimization problem over the space of matrices that is extremely challenging to solve exactly. It generalizes the subgraph isomorphism problem and all existing network alignment instances, which correspond to the case where $m=1$.

\section{Related work \textit{\&} Background}
\label{sec:background}

This section is meant to position our work within the broader context of network alignment methods. For recent surveys that explore these dimensions, see~\cite{Elmsallati-2016-network-alignment}. We are considering global network alignment which seeks a single matching that applies to the entire network instead of multiple local alignments between the networks. The global network alignment problem is almost always stated as an attempt to maximize the number of overlapped edges combined with maximizing domain-specific notions of apriori known vertex similarity. For instance, in protein-protein interaction networks, this notion of similarity is the similarity score of the respective genetic sequences~\cite{singh2008-isorank-multi}. In ontology alignment, this may come from a textual measure on the ontology label. The hardest network alignment problems lack these \emph{hints} about how to align the networks. Our method falls into this hard case where we assume no apriori knowledge of the alignment. Another important feature is scoring the resulting match when the networks have substantially different sizes and edge densities. This aspect has been explored most thoroughly in terms of biological networks~\cite{Vijayan-2015-magna++}. 

There are two prominent classes of methods: (i) embedding and (ii) integer optimization relaxations and heuristics. Embedding methods seek to compute a feature vector for each node of $A$ and $B$ independently. They then use relationships between the feature vectors to generate the matching. Using eigenvectors is common for this task in pattern recognition~\cite{knossow2009-laplacian}. Recent methods have proposed graphlet counts~\cite{Kuchaiev-2010-topological} and eigenvector histograms~\cite{Patro-2012-ghost}. One closely related method to our multimodal proposal is to generate the embeddings based on node and edge types~\cite{fraikin2007-graph-matching}. Obtaining these feature vectors can be extremely expensive in terms of computation, but results in memory efficient methods. 

Our method is most strongly related to the integer optimization and relaxation framework that seeks to maximize functions related to~\eqref{eq:mm-align} in the case of one mode, perhaps with an additional term reflecting the vertex similarity. This problem is a specific instance of a quadratic assignment problem~\cite{Burkard-2012-assignment-problems}. IsoRank was one of the first methods in this class~\cite{singh2008-isorank-multi}.  Subsequent techniques include belief propagation methods~\cite{Bayati-2013-netalign}, Lagrangian relaxations~\cite{klau2009-network-alignment}, spectral methods~\cite{Feizi-2016-spectral-net-align}, and tensor eigenvectors for motif-alignment~\cite{Mohammadi-2016-tame}. Essentially, all of these methods store a dense, real-valued heuristic matrix $\mY$ of size $|V_A| \times |V_B|$. This needs memory that is \emph{quadratic} in the size of the networks and greatly limits scalability. One of the most scalable methods results from a crucial analytical insight into the structure of the IsoRank heuristic. The resulting method -- network similarity decomposition~\cite{Mohammadi-2016-tame} -- can be considered a hybrid of embedding and integer optimization. Our results show how these ideas enable a seamless generalization to multimodal networks and we return to explain in more depth in the next section.


\section{Multimodal Similarity Decomposition: A~multi-modal~generalization~of~IsoRank~and~the~Network~Similarity~Decomposition}
\label{sec:msd}
We now present our approach to compute a multimodal network alignment. The high level idea is that we are going to run the IsoRank method to align the multimodal adjacency matrices $\mM$ and $\mN$ corresponding to multimodal networks $A$ and $B$. Doing so will involve a number of new insights about how the methods will behave on multimodal networks and exploiting the structure of the methods. The result of this section is a heuristic solution matrix $\YMN$ that we describe how to turn into a multimodal alignment in Section~\ref{sec:matching}. We begin by reviewing IsoRank~\cite{singh2008-isorank-multi} and the network similarity decomposition~\cite{Kollias-2011-netalign}. We assume the networks are undirected in this derivations, but note that directed generalizations are possible.

\subsection{IsoRank \textit{\&} the network similarity decomposition}

Consider aligning two standard networks $A$ and $B$ with adjacency matrices $\mA$ and $\mB$. The IsoRank method generates a heuristic matrix $\mY$ where $Y(v,v')$ is large if $v$ seems like it should match with $v'$. To do this, IsoRank computes $\mY$ as the solution of a PageRank problem on the product graph of $A$ and $B$. In a small surprise, this is a well-motivated idea. For more details about this, we refer the reader to~\cite{Bayati-2013-netalign}.  
We can express this PageRank problem in terms of the final matrix $\mY$ where $\mP$ is the degree-normalized matrix for network $A$, $P_{ij} = A_{ij}/d_j$, $\mQ$ is the degree-normalized matrix for $\mB$, and $S(v,v')$ is the apriori similarity of node $v$ in $A$ and $v'$ in $B$: 
\begin{equation}
	\label{eq:isorank}
\mY = \alpha \mP \mY \mQ^T 	+ (1-\alpha) \mS
\end{equation}
One small note: we assume in these equations that $A$ and $B$ have no nodes with zero outdegree. If there are such nodes, then we need to normalize $\mY$ to be a probability distribution.
When there is no apriori similarity known, then all entries $\mS$ are the constant $|V_A|^{-1} |V_B|^{-1}$, which reflects uniform similarity.

The insight in the network similarity decomposition (NSD) is that if $\mS$ is rank 1 -- as it would be in the case of uniform similarity -- then running $t$ iterations of the power method to compute PageRank will result in a rank $t+1$ matrix $\mY$. Moreover, the low-rank factorization of $\mY$ can be easily computed by running the power method for PageRank on network $A$ and network $B$ \emph{independently}. (In this sense, this is similar to an embedding method, because we compute the PageRank iterates of each network separately, then combine them.) This is easy to see when $\mS = \vu \vv^T$ then $t$ iterations of the power-method for \eqref{eq:isorank} produce: 
\[ \mY^{(t)} = (1-\alpha) \sum_{k=0}^{t-1} \alpha^k \mP^k \vu \vv^T (\mQ^T)^k + \alpha^t \mP^t \vu \vv^T (\mQ^T)^t\]
(This matrix is a sum of $t+1$ rank-1 matrices of the form $\texttt{scalar} \cdot (\mP^t \vu) (\mQ^t \vv)^T$.)
In this way, a low-rank factorization of the IsoRank solution can be computed extremely efficiently. Run $t$ steps of the power method for PageRank on network $A$ and $B$ separately and store the iterations (we will be more precise when we explain our method for multimodal networks next). Reasonable values of $t$ are between $5$ and $25$ as the PageRank series converges extremely quickly; for $\alpha$ we use $\alpha = 0.9$. 

\subsection{Our Multimodal Similarity Decomposition}
\label{sec:msd-detail} 
Our goals with the multimodal similarity decomposition and multimodal alignment closely mirror those of NSD and IsoRank. In our case, we want to allow alignment information to \emph{flow} between modes of the network to reinforce alignments between vertices that are present in various modes. This intuition suggests that the IsoRank scores for the alignment of multimodal adjacency matrices will be an effective heuristic for our problem. This will flow alignment information between modes via the cross-modal edges $\mC_{ij}$. 

Since we assume that the alignment between modes of the network are already known, we can strengthen our heuristic by customizing the matrix $\mS$. Recall that $S_{ij}$ is the apriori similarity of node $i$ in $A$ and $j$ in $B$. The multimodal adjacency matrix has $m$ copies of each vertex identifier -- one for each vertex in each of the $m$ modes. Thus, we use the matrix 
\begin{equation} \label{eq:mm-sim}
\mS = \bmat{ \gamma \cdot \texttt{ones} & 0 & \ldots & 0 \\
	            \vdots & \ddots & \ddots & \vdots \\
	            0 & \ldots & 0 & \gamma \cdot \texttt{ones} }, 
\end{equation}	        
where \texttt{ones} is a matrix of size $|V_A| \times |V_B|$ and $\gamma$ is chosen such that all entries of $\mS$ sum to 1 (that is, $\gamma = m^{-1} |V_A|^{-1} |V_B|^{-1}$). This choice of $\mS$ corresponds to a rank-$m$ matrix, where each mode constitutes a single rank-1 factor suggesting that nodes in mode 1 in multimodal network A should correspond to nodes in mode 1 of multimodal network B. 

Let $\mM$ and $\mN$ be the multimodal adjacency matrices for multimodal networks $A$ and $B$, respectively. Let $\PM$ and $\QN$ correspond to degree-normalization (by columns) of $\mM$ and $\mN$.  Thus, the heuristic $\YMN$ we compute solves: 
\begin{equation} \label{eq:mm-iso}
\YMN^{} = \alpha \PM^{} \YMN^{} \QN^T + (1-\alpha) \mS
\end{equation}  
where $\mS$ is the multimodal similarity \eqref{eq:mm-sim}. The result $\YMN$ gives a score for aligning each node in each mode of $A$ to each node in each mode $B$. This is not the type of output we want, so in subsequent sections, we discuss how to turn this output into the type of matching we expect. 

At this point, we can utilize the same methodology underlying the network similarity decomposition to efficiently compute a low-rank decomposition of $\YMN$. Note that the matrix $\mS$ is rank-$m$. Thus, we can compute $t$ iterations of the power-method to solve \eqref{eq:mm-iso} and the result is a rank $(m) \cdot (t+1)$ matrix $\YMN$ following the exact same argument as expression for the network similarity decomposition.  The pseudocode is shown in Figure~\ref{fig:pseudocode}. One small detail is that we assumed in the derivation that $\PM$ and $\QN$ have no empty rows or columns. If this is not the case, then each iterate needs to be normalized to be a probability distribution as in the pseudocode. 
\begin{figure}
\begin{pseudocode}
Input: Multimodal networks $A$, $B$, $\alpha$, iteration $t$
Output: Matrices $\mU$ and $\mV$ such that $\smash{\YMN^{(t)} = \mU \mV^T}$
$\mU = $ PageRankPowers(A,$\alpha$,$t$)
$\mV = $ PageRankPowers(B,$\alpha$,$t$)

PageRankPowers(A,$\alpha$,$t$)
Set $\mM$ to be the multimodal adjacency of $A$.
Set $\mP$ to be the column normalized matrix $\mM$. 
Allocate Z as a $|V| \times (t+1) \times m$ array 
for k=1 to m (the number of modes)
  for all i=1 to $|V|m$
    Z[i,0,k] = $\textstyle\begin{cases} (\sqrt{m}|V|)^{-1}  & i \text{ is an index in mode } k \\ 0 & \text{otherwise} \end{cases}$
  for j=1:t
	Z[:,j,k] = normalize($\mP$ Z[:,j-1,k])
	Z[:,j-1,k] = $\sqrt{(1-\alpha)\alpha^{j-1}}$ Z[:,j-1,k]
  end
  Z[:,t,k] = $\sqrt{\alpha^t}$ Z[:,t,k]
end  
return $\mU$ as Z reshaped to $|V| \times (t+1)(m)$
\end{pseudocode}
	\vspace{-1em}
\caption{Pseudocode for the multimodal similarity decomposition: $m$ is the number of modes of the multimodal networks $A$ and $B$, and the column-normalization computes: $P_{ij} = M_{ij} / \sum_{\ell} M_{\ell,j}$. The \texttt{normalize} function divides a vector by its sum. The reshaping occurs by column such that $U(:,1) = Z(:,0,1)$. We assume the reshaping happens the same way on $\mU$ and $\mV$ to ensure the low-rank factors are aligned. The result $\mU$ and $\mV$ must still be matched as in Section~\ref{sec:matching} to solve~\eqref{eq:mm-align}.}
\label{fig:pseudocode} 
\end{figure}

\section{Matching algorithms \textit{\&} Resolving conflicts}
\label{sec:matching}
\label{matching-algs}

After the multimodal network decomposition (MSD), the next step is to produce a matching between the vertices. A standard technique here is to run a max-weight bipartite matching on the matrix $\mY$ for the network alignment case~\cite{singh2008-isorank-multi,Bayati-2013-netalign}. There are two issues with this step. The first is the size of the resulting matrix $\YMN$ can be prohibitively large if we realize it as a dense matrix. This is especially true if we have a large multimodal network with many modes. A problem with $100$ modes and $5000$ vertices would produce a $500,000\times 500,000$ dense matrix whereas the low-rank factors for $8$ iterations would require two $500,000 \times 900$ matrices. Thus, we investigate memory-efficient matching routines that deal with $\YMN$ via its low-rank decomposition. The second issue is that this output is a matching on the rows of the multimodal adjacency matrix and not a matching on the vertex set of the multimodal network. Aggregating them can result in conflicting matches and we describe a few ways to resolve these. 

\subsection{Approximate matching via low-rank factors}
Bipartite matching is a common subroutine that has a well-known polynomial time algorithm. The problem with this algorithm is that it requires the matrix of weights associated with each edge of the bipartite graph. For our case, this matrix is large and dense, rendering the algorithm infeasible due to the memory required. However, note that the bottleneck is the density of the matrix. Bipartite matching is computationally tractable with \emph{sparse} matrices of exactly the same size because the algorithms utilize the sparsity to lower the memory requirement. For the purposes of this section, let $\mY = \mU\mV^T$ be the matrix we wish to run a bipartite matching on and $\mU$ be $m \times k$. We note that the low-rank factors from MSD are non-negative, so we assume $\mU, \mV \ge 0$. 

We begin with an extremely idealized case. If $k=1$, then $\mY = \vu \vv^T$ is rank-1, and it is extremely easy to compute a maximum weight matching. This only involves \emph{sorting} the vectors $\vu$ and $\vv$ in decreasing order and taking the largest feasible set of components. (This fact is a simple corollary of the rearrangement inequality, so we omit a formal statement and proof.) 

For the remainder of this derivation, we'll need one additional bit of notation: the matrix inner product $\mA \bullet \mB = \sum_{ij} A_{ij} B_{ij}$. Recall that a matching can be expressed as a matrix where $X_{ij} = 1$ if $i$ is matched to $j$ and all other entries are 0. Given any matching $\mX$ the weight of the matching is $\mX \bullet \mY$. We now show that using the \emph{best matching} from each rank-1 factor of $\mY$ gives a $1/k$-approximation when $\mY$ has rank-$k$. 

\begin{theorem}
Let $\mY = \mU \mV^T$ be a rank-$k$ non-negative decomposition for $\mY$. Let $\mX_i$ be the maximum weight matching corresponding to the rank-1 factor $\vu_i^{} \vv_i^T$, and let $f_i = \mX_i \bullet (\vu_i\vv_i^T)$ be its weight. If $f^*$ is the largest value of $f_i$, then $f^*$ is a $1/k$-approximation to the maximum weight matching on $\mY$.  
\end{theorem}
\begin{proof} 
Let $\mX^*$ be any optimal maximum weight matching. Then 
\[\underbrace{\vphantom{\textstyle\sum_{k=1}^n}\mX^* \bullet \mY}_{\substack{\text{optimal matching}\\ \text{ weight}}} = \underbrace{\textstyle\sum_{i=1}^k (\mX^* \bullet \vu_i \vv_i^T) \le \textstyle\sum_{i=1}^k f_i}_{\text{$f_i$ is optimal for $\vu_i \vv_i^T$ }} \le k f^*. \]
\end{proof} 

While this provides a simple approximation bound, the value of $k$ for our experiments is often around $1000$, rendering it rather pointless in terms of theory. However, we find that a single factor often produces exceptionally good approximation factors -- far beyond what the theory from this result would show. We plan to continue investigating this problem and conjecture that the structure of the MSD vectors themselves may suggest a tighter approximation.  

In practice, there are a few additional improvements we can make. Two of these apply to general problems and also produce a $1/k$ approximation. The third improvement is specific to the multimodal network alignment objective. We also discuss a few other ideas suggested in the literature. 

\paragraph{Maximum weight 1/k.} In the first variant, we compute the weight of the matching $\mX_i$ in the full matrix $\mY$ via $\mY \bullet \mX_i$. The runtime of this computation is $nk$ and it can be done entirely with the low-rank factors. Then we select the largest weight among matchings $\mX_i$. This is still a $1/k$ approximation because $\mY \bullet \mX_i \ge (\vu_i \vv_i^T) \bullet \mX_i$ due to the non-negativity.

\paragraph{Union of matchings.} While we cannot always expect to solve the dense bipartite matching problem, we are able to solve sparse bipartite matching problems on the size of the vertex sets easily. In this case, we can take the union of edges produced by the matchings $\mX_1, \ldots, \mX_k$ and find the best bipartite matching in this sparse subset of the entries of $\mY$. Again, this can only make our approximation better and so we still get $1/k$.

\paragraph{Maximum overlap.} The actual goal of our problem is to achieve the best alignment between two multimodal networks. Thus, rather than compute the \emph{weight} of the matching in the heuristic $\mY$, we directly evaluate the \emph{overlap} produced by each of the matchings $\mX_i$ between the multimodal adjacency matrices. (Note this is slightly different from the matching we will consider after we resolve conflicts below.) This picks the matching $\mX_i$ that  corresponds with the true objective functions. 

\paragraph{Other possibilities.} In ref.~\cite{Kollias-2014-parallel-nsd}, they suggest computing $\mY$ a row at a time and retaining the largest $r$ entries as edges of a graph. This forms a sparse graph that can be used in a parallel maximum weight bipartite matching.  Alternatively, it is feasible to run a greedy $1/2$-approximate matching. We found both of these underperformed our ideas in terms of final quality and were no faster. 

\paragraph{Runtime \textit{\&} Memory} All of these schemes keep the matrix $\mY$ stored as factors $\mU$ and $\mV$. Each matching takes $O(n)$ storage, so storing $k$ matches requires $O(nk)$ memory, and this exactly mirrors the storage of the factors $\mU$ and $\mV$ themselves. We provide a comparison of runtime in Table~\ref{tab:runtime}. 
\begin{margintable}
	\vspace*{-0.5\baselineskip}
	\caption{We compare runtimes of approximate bipartite matching where $\mU$ and $\mV$ are $n \times k$ for simplicity; all methods use $O(nk)$ memory. Also, $\text{BM}(E)$ is the runtime for bipartite matching on $n$ nodes with $E$ total entries, and $|E|$ is the number of edges in the original networks. } 
	\label{tab:runtime}
	\noindent\begin{tabularx}{\linewidth}{@{}lX@{}}	
		\toprule
		Method & Runtime  \\
		\midrule
		Simple $1/k$ & $nk \log n$ \\
		Max Weight $1/k$ & $nk \log n + nk^2$ \\
		Union $1/k$ & $nk \log n + nk^2 + \text{BM}(nk)$ \\		
		Max Overlap & $nk \log n + |E|$ \\
		$r$-Sparsified & $n^2k + \text{BM}(nr)$ \\
		Greedy $1/2$ & $n^3k$ \\ 
		\bottomrule
	\end{tabularx}
	\vspace{-\baselineskip}
\end{margintable}

\paragraph{Practical considerations.} All of the approximations here are extremely fast to compute for values of $n$ up to a few million. So in our codes, we will usually compute multiple matchings and choose the best based on the overall alignment scores after we resolve conflicts as discussed next. 

	\vspace*{-0.5\baselineskip}
\subsection{Resolving conflicts}
So far, we have presented methods to derive a matching on a low-rank approximation. In the context of multimodal networks, this matches at the level of the row and column identifiers of the multimodal adjacency matrices. This is a problem because each vertex id in the multimodal network occurs at $m$ different rows of the multimodal adjacency matrix. For instance, if multimodal network $A$ consists of $3$ modes, vertex $1$ can be matched to three different vertices, one for each mode!  

To resolve these conflicts, we have two methods and choose the best among them based on the overlap of the alignment they produce. The first method is to greedily examine the multimodal matching, project each match to nodes in $V_A$ and $V_B$, and accept it if it's still feasible in the projected alignment. The second method is to project the multimodal matchings to one bipartite network on $V_A$ and $V_B$. Put another way, we take the union of all matches induced by the single multimodal matching. Each edge is weighted by the value in $\YMN$ (and duplicates are summed) that arose from the matching.  Once we project all the matchings, we use a bipartite matching to obtain the final matching.
\section{Baseline methods}
\label{sec:baseline}

Before we describe our experiments, we explain how the multimodal alignment problem can be addressed using existing methods for network alignment. These methods do not take advantage of the multimodality of the data to guide the alignment, but nevertheless provide strong baselines. The first baseline is to smash the networks together and look at an alignment of the unimodal networks with edge set $\cup_{k=1}^m E_A^{(k)}$ and $\cup_{k=1}^m E_B^{(k)}$. We can use any existing network alignment tool for this task. In the following experiments, we use Klau's relaxation~\cite{klau2009-network-alignment} because it had the highest performance among a number of methods we evaluated (see online codes). In addition to aligning the smashed unimodal networks, we can also compute a pairwise alignment between each mode. Thus, we align the edges $E_A^{(k)}$ to $E_B^{(k)}$, which provides another matching. Note that each alignment produces exactly the type of output we need and there is no need to resolve conflicts as we saw in the multimodal case. We can then pick the \emph{best} among these $m+1$ alignments (1 alignment for each mode and 1 alignment for the smashed network) based on the total overlap they induce in the full multimodal network.  We call these methods \emph{pairwise alignment} rather than \emph{multimodal alignment}.

\section{Synthetic experiments}
\label{sec:synthetic}

\begin{fullwidthfigure}[t]
	\parbox{0.33\linewidth}{\centering ~~~~~~~~\textbf{MSD}}%
	\parbox{0.33\linewidth}{\centering ~~~~~~~~\textbf{Best of Pairwise}}%
	\parbox{0.33\linewidth}{\centering ~~~~~~~~\textbf{Difference}}
	
	\includegraphics[width=0.33\linewidth]{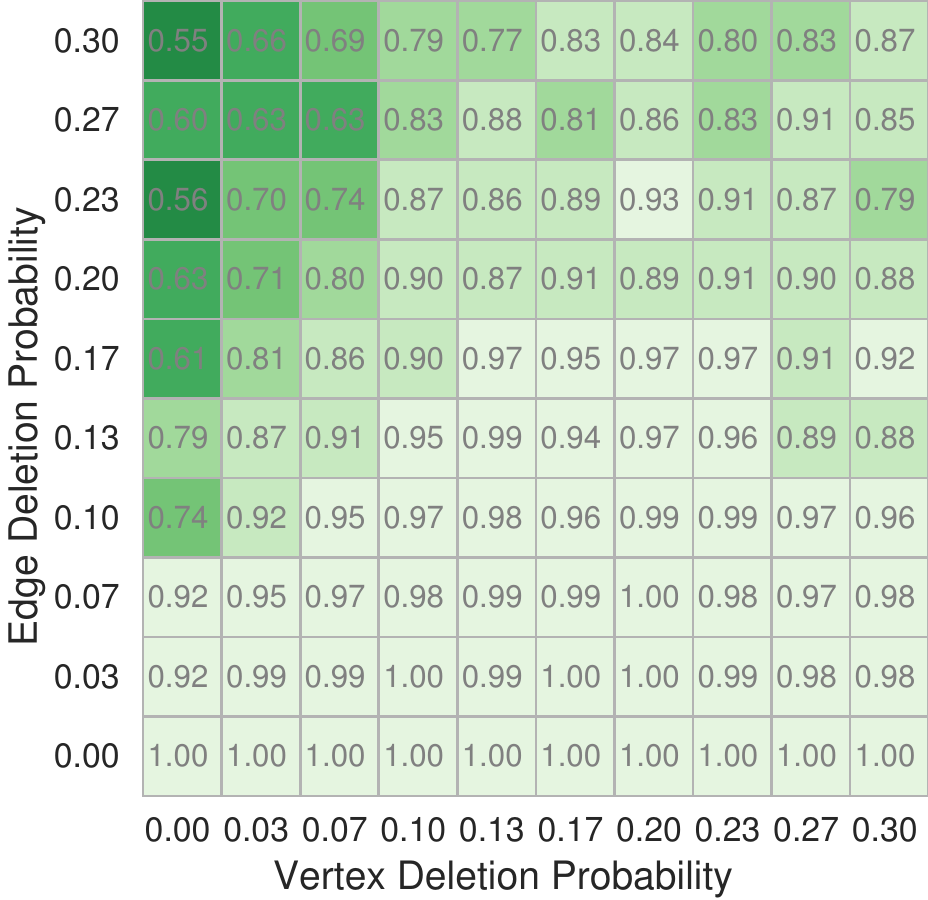}%
	\includegraphics[width=0.33\linewidth]{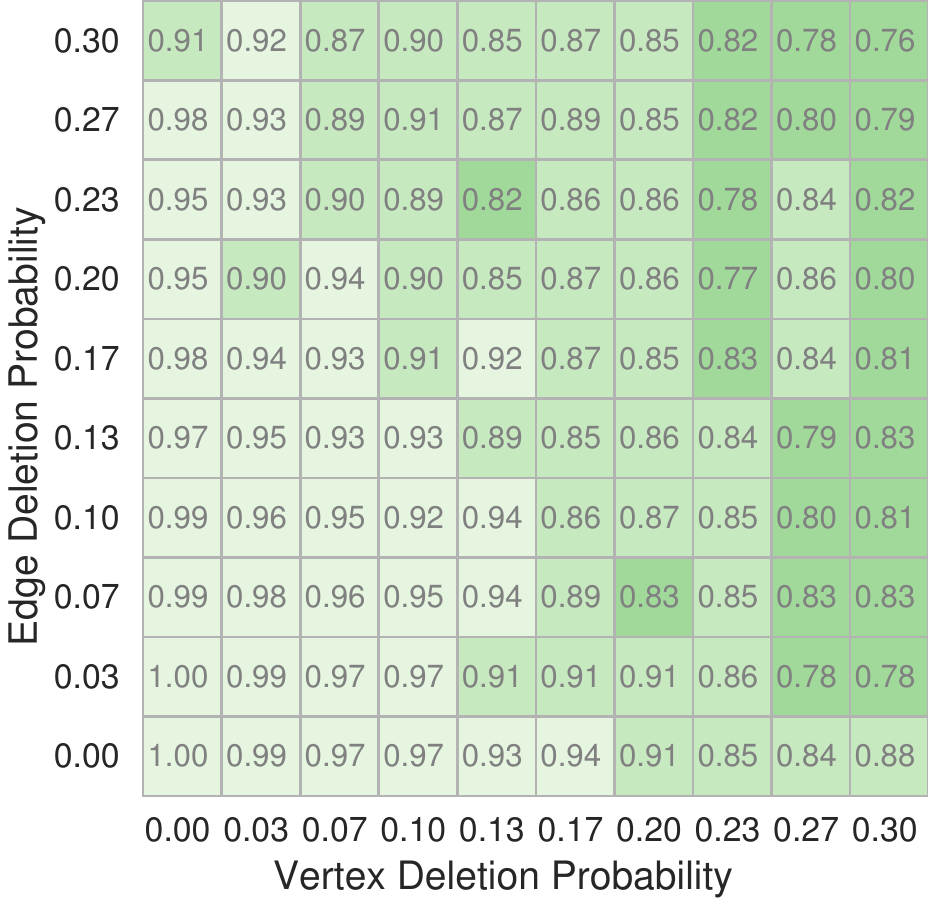}%
	\includegraphics[width=0.33\linewidth]{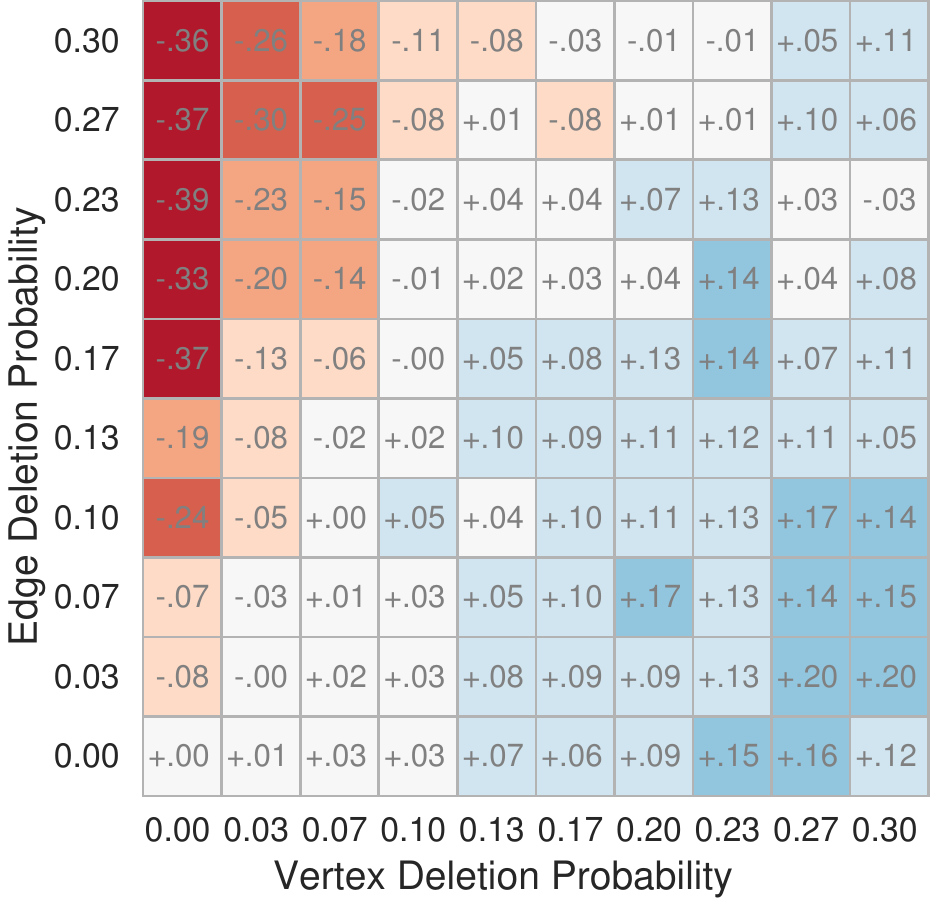}
	
	\caption{At left, the recovery results for our multimodal similarity decomposition (MSD) show excellent performance even in the high vertex deletion regime. In the middle, the best of any pairwise recovery result shows good results in the high edge deletion, but worse results with high vertex deletion. At right, the difference shows blue when MSD outperforms the best of the pairwise alignments by at least 5\% recovery. }
	\label{fig:recovery}
\end{fullwidthfigure}

\begin{figure}
	\parbox{0.5\linewidth}{\makebox[0pt][l]{\footnotesize~~~~~Vertex Del.~0.1, Edge Del.~0.2}}%
	\parbox{0.5\linewidth}{\makebox[0pt][l]{\footnotesize~~~~~Vertex Del.~0.2, Edge Del.~0.1}}
	
	\includegraphics[width=0.5\linewidth]{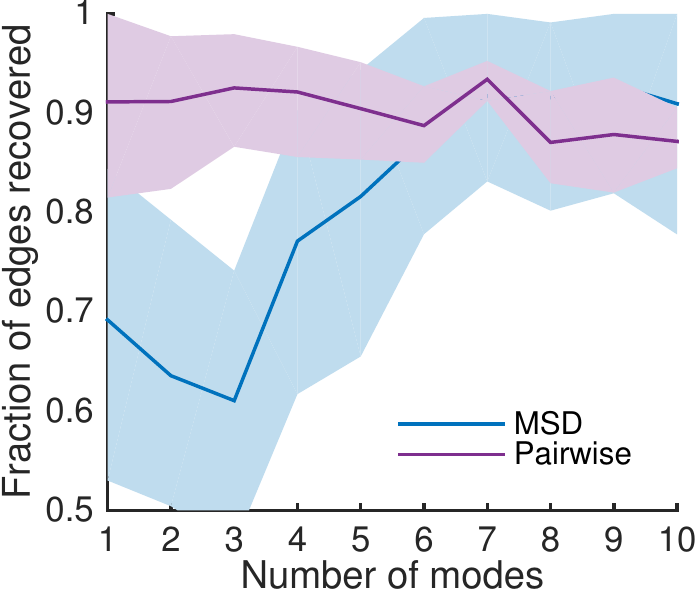}%
	\includegraphics[width=0.5\linewidth]{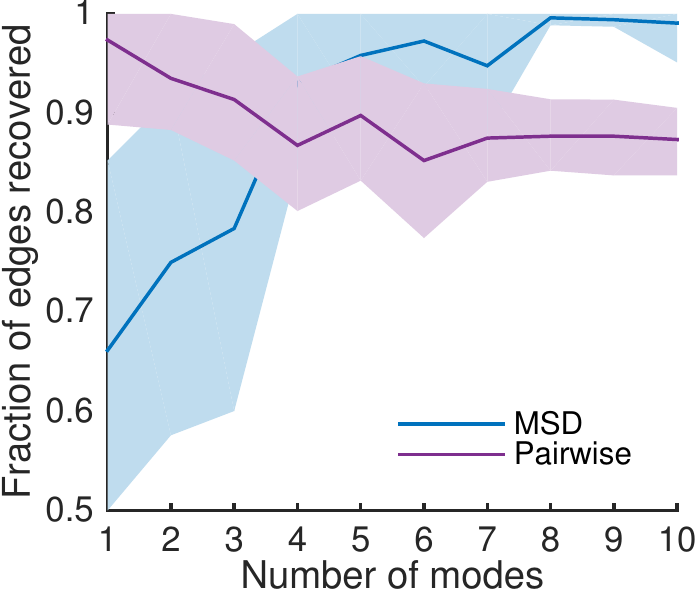}
	
	\caption{At left, we show how the recovery rate changes as we add modes when $p=0.1$ and $q=0.2$; at right we see $p=0.2$ and $q=0.1$. The performance of our multimodal alignment increases as we add modes. The line is the mean over 50 trials and the bands show the 10\% and 90\% percentiles.}
	\label{fig:adding_modes}
\end{figure}

In our first experiment, we want to understand when multimodal alignment results in better alignment compared with using a strong algorithm for the pairwise case. We use our multimodal method with $\alpha = 0.9$, $10$ iterations, and then compute the dense matrix $\YMN$ to use with bipartite matching because this test case is sufficiently small. We compare it against Klau's method~\cite{klau2009-network-alignment} as a pairwise aligner as in Section~\ref{sec:baseline}. Klau's algorithm takes significantly longer than MSD. 

We generate a multimodal network alignment problem with $36$ nodes and $6$ modes as follows. We first create a 12-node Erd\H{o}s-R\'{e}nyi graph with average degree 3, and then join copies into a single 36-node network. The purpose of combining the networks is to put a degree of symmetry into the graph to make the alignment harder. (For more details about the networks construction refer to the code associated with this paper). This is our reference graph.   Next, we generate a mode by randomly deleting vertices with probability $p$ and randomly deleting edges with probability $q/2$. Then we generate two instances of this modal graph with deleted edges again with probability $q/2$. One instance goes to multimodal network $A$ and the other goes to multimodal network $B$. At the end, we have two multimodal networks $A$ and $B$ where each mode shares a number of relationships. 

In figure~\ref{fig:recovery}, we show the fraction of edges aligned using a method over the total number of edges in the networks  
as we vary $p$ and $q$ for both MSD and the best of pairwise alignments (Section~\ref{sec:baseline}).  We do not use node-based recovery due to many nodes of degree one, which cannot be resolved.
The results are the mean over $50$ trials, and they show a large regime where multimodal alignment is superior to any pairwise alignment. Specifically, when the vertex deletion probability $p$ is large, the multimodal alignment is the only method to accurately align the networks. When edge deletion is high, smashing the networks effectively reconstitutes the original network and we see the superior performance of Klau's pairwise method. 

Next, in Figure~\ref{fig:adding_modes}, we show how the behavior of the methods change as we vary the number of modes of the multimodal network. As expected, the performance of our method increases with additional modes, whereas the performance of the pairwise method is consistent (or even slightly decreasing). 

\section{De-anonymizing transportation datasets}
\label{sec:airports}
In this section we present a case-study with de-anonymizing a publicly released dataset of airlines and airports, where we treat each airline as a mode. We show in this section that only our proposed multimodal alignment method can fully de-anonymize the network and align each edge. We run these experiments with $\alpha = 0.9$, $10$ and illustrate how various low-rank matching procedures perform. The anonymized network is the European air network from ref.~\cite{Nicosia-2015-multiplex}, which was released with anonymized airport identifiers but with all airline identifiers. The data originally came from the OpenFlights repository (\url{http://openflights.org}). We wanted to take the original data and use it to restore identifiers to the anonymized data. We selected the May 2013 release based on the publication dates.  These networks have $594$ airports and $175$ airlines and $6468$ edges, respectively. Our goal is to align the multimodal network such that each edge is matched. 

\subsection{Performance at de-anonymization}
We consider an alignment to de-anonymize the network if it is able to overlap all 6468 edges through the matched vertex identifiers. In Table~\ref{tab:airports}, we show the overlap size for our multimodal methods and a number of different low-rank matching techniques from Section~\ref{sec:matching}.  This time we are not able to run the full bipartite matching unless we use a computer with $512 GB$ of RAM and so we don't report those results. The best result from the baseline pairwise alignment methods miss 70 edges. All of our low-rank matching approximations achieve the full overlap score.

\subsection{Which modes help multimodal alignment?}
The final experiment seeks to understand \emph{which} modes have the highest impact on the de-anonymization performance. In many applications, collecting additional data incurs a cost and this experiment is designed to provide insight as far as what type of multimodal data would be most helpful to collect. Our goal is to use only a subset of the $175$ modes to produce an alignment, and then compute the overlap that results from using this alignment on all $175$ modes. We sort the modes based on a number of graph theoretic measures to produce interesting subsets. The per-mode measures are: edge count, unique vertex count, average degree, triangle count, density. Figure~\ref{fig:one-mode-by-mode} shows the results for these measures along with a few random orderings of the modes. The goal is to align the highest fraction of edges using the fewest modes. Most of the graph theoretic measures have similar performance, which suggests that any could be a proxy for which data to collect. Density is a notable outlier as some of the modes consist of a single edge, resulting in a high density, but little information about the global alignment.

\begin{table}[t]
	\caption{The multimodal airport networks have 6468 edges. All the multimodal methods are able to de-anonymize this network. The best pairwise results come from smashing the multimodal structure and aligning the resulting networks, which misses 70 edges compared to the multimodal alignments and takes more time. } 
	\label{tab:airports}
	\footnotesize
	\noindent\begin{tabularx}{\linewidth}{llX}	
		\toprule
		Type & Method & Overlap  \\
		\midrule
		Multimodal & MSD \famp\ Max Weight $1/k$ & $6468$ \\
		Multimodal & MSD \famp\ Union $1/k$ & $6468$ \\
		Multimodal & MSD \famp\ Max Overlap & $6468$ \\ 
		\midrule 
		Pairwise & smashed				& $6398$ \\
		Pairwise & best mode			& $3127$ \\
		\bottomrule
	\end{tabularx}
\end{table}

\begin{figure}[t]
	\centering
		\vspace*{-0.5\baselineskip}
	\includegraphics[width=0.7\linewidth]{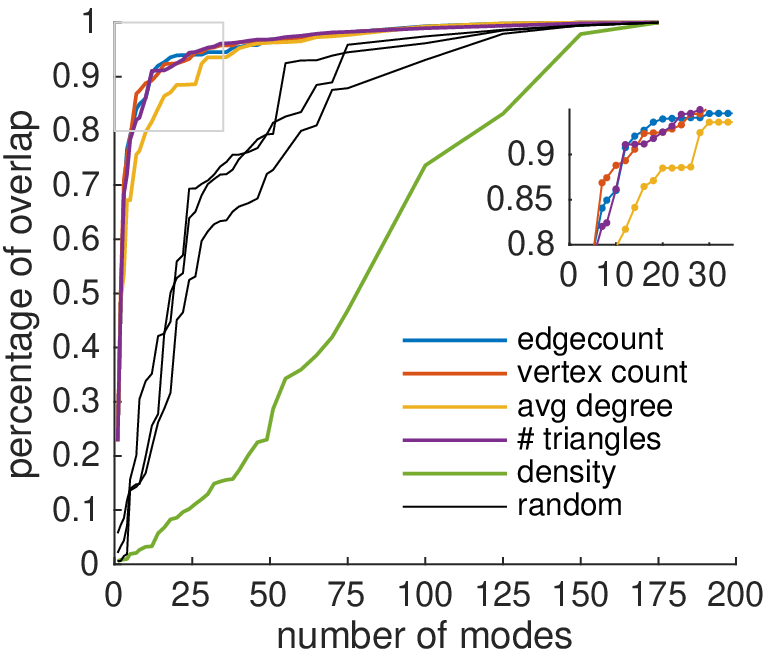}
	\vspace{-1em}	
	\caption{As we use more modes to produce an alignment for the  airport data, the performance increases. This experiment shows that using modes that have many edges (edge count), touch many unique vertices (vertex count), have a high average degree (avg degree), or have high triangle count (\# triangles) all outperform random selections. Alternatively, density does not work as discussed in the main text.}
	\label{fig:one-mode-by-mode}	
	\vspace{-\baselineskip}
\end{figure}

\vspace{-1ex}
\section{Conclusion \famp\ Future Work}

This paper demonstrates an advantage to using multimodal features of data: it makes alignment problems easier. Here, we compared our multimodal network decomposition (MSD) against a carefully engineered method in pairwise alignment (Klau's). We also did all of the pairwise experiments reported here with the IsoRank method instead of Klau's and the results were uniformly worse as far as the alignment quality. However, when we use the multimodal extension to IsoRank we have proposed, accurate alignments are easy to obtain. For instance, note that in Figure~\ref{fig:adding_modes}, MSD used the additional modes to improve the alignment, whereas the pairwise method did not.

Although we have not captured runtimes precisely in this manuscript, let us note that Klau's method takes a great deal more computational effort than our approximation, and results in a extremely high quality pairwise solution for any mode or the smashed networks. To give some sense of the difference in effort required, the results on the synthetic experiment for Figure~\ref{fig:recovery} took \emph{hours} to compute with Klau's method whereas MSD and all the matching took a few minutes. 

In future work, we plan to explore using this type of multimodal alignment in the context of protein and gene relationships. We believe that using this new methodology -- combined with particular domain specific adaptations -- will result in new biological insights. 

\section*{Acknowledgments}
Supported by DARPA SIMPLEX and NSF award CCF-1149756, IIS-1422918, IIS-1546488, CCF-093937 and the Sloan Foundation.

\newpage
\begin{fullwidth}
\bibliographystyle{dgleich-bib}
\scriptsize
\bibliography{99-refs}

\end{fullwidth}

%
%

\end{document}